\documentclass[letterpaper,12pt]{article}
\usepackage{natbib}
\usepackage{har2nat}
\usepackage[dvips]{graphicx}
\usepackage{amssymb}
\usepackage{amsmath}
\usepackage{amsthm}
\usepackage{subeqnarray}
\usepackage{setspace}

\usepackage{url}

\newcommand \be{\begin{equation}}
\newcommand \bea{\begin{eqnarray}}
\newcommand \ee{\end{equation}}
\newcommand \eea{\end{eqnarray}}

\newcommand{\E}{{\rm E}}
\newcommand{\Var}{{\rm Var}~}

\renewcommand{\(}{\left(}
\renewcommand{\)}{\right)}
\renewcommand{\[}{\left[}
\renewcommand{\]}{\right]}
\renewcommand{\epsilon}{\varepsilon}

\theoremstyle{plain}

\newtheorem{proposition}{Proposition}
\newtheorem{lemma}{Lemma}
\newtheorem{corollary}{Corollary}

\theoremstyle{definition}

\newtheorem{assumption}{Assumption}
\newtheorem{remark}{Remark}

\begin{document}

\title{Zipf's law and maximum sustainable growth
\thanks{Y. Malevergne acknowledges financial support from the  French National Research Agency (ANR) through the ``Entreprises'' Program (Project HYPERCROIS n$^o$ ANR-07-ENTR-008).}
}
\thispagestyle{empty}

\author{Y. Malevergne$^{1,2,3}$, A. Saichev$^{3,4}$ and D. Sornette$^{3,5}$\\
{\small $^1$ Universit\'{e} de Lyon - Universit\'{e} de Saint-Etienne -- Coactis E.A. 4161, France}\\
$^2$ {\small EMLYON Business School -- Cefra, France}\\
$^3$ {\small ETH Zurich -- Department of Management, Technology and Economics, Switzerland}\\
$^4$ {\small Nizhny Novgorod State University -- Department of Mathematics, Russia}\\
$^5$ {\small Swiss Finance Institute, Switzerland}\\
{\small e-mails: ymalevergne@ethz.ch, saichev@hotmail.com and dsornette@ethz.ch}
}

\date{}
\maketitle

\begin{abstract}
Zipf's law states that the number of firms with size greater than $S$ is
inversely proportional to $S$. Most explanations start with Gibrat's
rule of proportional growth but require additional
constraints. We show that Gibrat's rule, at all
firm levels, yields Zipf's law under a balance condition between the
effective growth rate of incumbent firms (which includes their 
possible demise) and
the growth rate of investments in entrant firms. Remarkably, Zipf's law
is the signature of the long-term optimal allocation of
resources that ensures the maximum sustainable growth rate of an
economy.
\end{abstract}
\vspace{2cm}

\vspace{0.5cm}


\clearpage
\pagenumbering{arabic}

\centerline{\LARGE Zipf's law and maximum sustainable growth}

\vskip 5cm

\begin{abstract}
Zipf's law states that the number of firms with size greater than $S$ is
inversely proportional to $S$. Most explanations start with Gibrat's
rule of proportional growth but require additional
constraints. We show that Gibrat's rule, at all
firm levels, yields Zipf's law under a balance condition between the
effective growth rate of incumbent firms (which includes their 
possible demise) and
the growth rate of investments in entrant firms. Remarkably, Zipf's law
is the signature of the long-term optimal allocation of
resources that ensures the maximum sustainable growth rate of an
economy.
\end{abstract}
\vspace{2cm}

\vspace{0.5cm}

\clearpage

\setstretch{1.5}
\section{Introduction}

The relevance of power law distributions of firm sizes to help understand 
firm and economic growth has been recognized early, for instance by
\citeasnoun{Schumpeter1934}, who proposed that there might be
important links between firm size distributions and firm growth. The
endogenous and exogenous processes and factors that combine 
to shape the distribution of firm sizes can be
expected to be at least partially revealed by the characteristics of the 
distribution of firm sizes.
The distribution of firm sizes has also attracted a great deal of
attention in the recent policy debate \cite[for instance]{Eurostat1998}, 
because it may influence job creation and destruction \cite{Davisetal96}, the
response of the economy to monetary shocks \cite{Gertler1994} and might even be
an important determinant of productivity growth at the macroeconomic level due
to the role of market structure \cite{Peretto1999,Pagano2003,Acsetal1999}.

This article presents a reduced form model that provides a generic explanation for
the ubiquitous stylized observation of power law distributions of firm sizes, and in 
particular of Zipf's law -- i.e., the fact that the fraction of firms of an economy
whose sizes $S$ are larger than $s$ is inversely proportional to $s$:
$\Pr(S>s)\sim s^{-m}$, with $m$ equal (or close) to $1$.
We consider an economy made of a large number of firms that are created
according to a random birth flow, disappear
when failing to remain above a viable size, go bankrupt when an operational fault strikes, and grow or
shrink stochastically at each time step proportionally to their current sizes (Gibrat law).

Our contribution to the ongoing debate on the shape of the distribution of
firms' sizes is to present a theory that encompasses previous approaches and
to derive Zipf's law as the result of the
combination of simple but realistic stochastic processes of firms' birth and
death together with Gibrat's law \cite{Gibrat1931}.  The main result of our approach is that Zipf's law 
is associated with a maximum sustainable growth of investments in the creation of new firms. 
In this respect, the size distribution of firms appears as a device to assess the efficiency and the sustainability of the resources allocation process of an economy.
Another interesting aspect of our framework is the analysis of deviations from the pure Zipf's law (case $m=1$)
under a variety of circumstances resulting from transient imbalances 
between the average growth rate of incumbent firms and the growth rate of investments in new entrant firms.
These deviations from
the pure Zipf's law have been documented for a variety of firm's size proxies 
(e.g.  sales, incomes, number of employees, or total assets), and reported values
for $m$ ranges from $0.8$ to $1.2$ \cite[among many
others]{IS77,Sutton97,Axtell01}. Our approach provides a framework for
identifying their possible (multiple) origins.

In the literature on the growth dynamics of business firms, a well established
tradition describes the change of the firm's size, over a given period of
time, as the cumulative effect of a number of different shocks originated by the
diverse accidents that affected the firm in that period  \cite[among
others]{Kalecki45,IS77,Steindl65,Sutton1998,Geroski1999}. This, together with
Gibrat's law of proportional growth, forms the starting point for various
attempts to explain Zipf's law. However, these attempts generally start with the implicit
or explicit assumption that the set of firms under consideration was born at the
same origin of time and live forever \cite{Gibrat1931,Gabaix99,RHW2007a,RHW2007b}. This approach is equivalent to considering
that the economy is made of only one single firm and that the distribution of
firm sizes reaches a steady-state if and only if the distribution of the size of
a single firm reaches a steady state. This latter assumption is counterfactual or, even
worse, non-falsifiable.

An alternative approach to model a stationary distribution of firm
sizes is to account for the fact that firms do not all appear at the same time
but are born according to a more or less regular flow of newly created firms, as suggested by the common sense\footnote{See \citeasnoun{Dunne1}, \citeasnoun{Reynolds94} or \citeasnoun{BD04}, among many others, for ``demographic'' studies on the populations of firms.}.
\citeasnoun{Simon55} was the first to address this question (see also \citeasnoun{IS77}). 
He proposed to modify Gibrat's model by
accounting for the entry of new firms over time as the overall industry grows.
He then obtained a steady-state distribution of firm sizes with a regularly varying
upper tail whose exponent $m$ goes to one from above, in the limit of a
vanishingly small probability that a new firm is created. This situation is not
quite relevant to explain empirical data, insofar as the convergence toward the steady-state is then
infinitely slow, as noted by \citeasnoun{Krugman96}. More recently,
\citeasnoun{Gabaix99} allowed for birth of new entities, with the probability to
create a new entity of a given size being proportional to the current fraction
of entities of that size and otherwise independent of time. In fact, this assumption
does not reflect the real dynamics of firms' creation. For
instance, \citeasnoun{Bartelsmanetal03} document that entrant firms have a
relatively small size compared with the more mature efficient size they  develop
as they grow.  It seems unrealistic to expect a non-zero
probability for the birth of a firm of very large size, say, of size comparable
to the largest capitalization currently in the market\footnote{We do not
consider spin-off's or M\&A (mergers and acquisitions).}. In this respect, \citeasnoun{Luttmer07}'s model is more realistic than Gabaix's, (who anyway models city sizes rather than firms) insofar as it considers that entrant firms adopt a scaled-down version of the technology of incumbent firms and therefore endogenously set the size of entrant firms as a fraction of the size of operating firms. In this article, we partly follow this view and consider that the size of entrant firms is smaller than the size of incumbent firms. But we depart from Luttmer's because the size of new entrants is not endogenously fixed in our model. We set this parameter exogenously for versatility reasons.

Another crucial ingredient characterizes our model.
The fact that firms can go bankrupt and disappear
from the economy is a crucial observation that is often neglected in models.  
Many firms are known to undergo
transient periods of decay which, when persistent,
may ultimately lead to their exit from business \cite{BD04,Knaup05,Brixy07,Bartelsmanetal03}.
\citeasnoun{Simon60} as well as \citeasnoun{Steindl65} have considered this
stylized fact within a generalization of \citeasnoun{Simon55} where the
decline of a firm and ultimately its exit occurs when its size reaches
zero. In \citeasnoun{Simon60}'s model, the rate of firms' exit exactly
compensates the flow of firms' births so that the economy is stationary and
the steady-state distribution of firm sizes exhibit the same upper tail behavior
as in \citeasnoun{Simon55}. In contrast, \citeasnoun{Steindl65} includes births
and deaths but within an industry with a growing number of firms. A steady-state
distribution is obtained whose tail follows a power law with an exponent that
depends on the net entry rate of new firms and on the average growth rate of
incumbent firms. Zipf's law is only recovered in the limit where the net entry
rate of new firms goes to zero. Both models rely on the existence of a minimum
size below which a firm runs out of business. This hypothesis corresponds to the
existence of a minimum efficient size below which a firm cannot operate, as is
well established in economic theory. However, there may be in general more than
one minimum size as the exit (death) level of a firm has no reason to be equal
to the size of a firm at birth. In the afore mentioned models, these two sizes
are assumed to be equal,  while there is {\em a priori} no reason for such an
assumption and empirical evidence {\it a contrario}. In our model, we 
allow for two different thresholds, the first one for the typical size of entrant firms and the second one for the exit level. This second level is assumed to be lower than the first one, even if recent evidence seems to suggest that firms might enter with a size less than their minimum efficient size \cite{AA2001} and then rapidly grow beyond this threshold in order to survive.

In addition to the exit of a firm resulting from its value decreasing below a
certain level, it sometimes happens that a firm encounters financial troubles
while its asset value is still fairly high. One could cite the striking examples of Enron
Corp. and Worldcom, whose market capitalization were supposedly high (actually the
result of inflated total asset value of about \$11 billion for Worldcom and
probably much higher for Enron) when they went bankrupt. More recently, since mid-2007
and over much of 2008, the cascade of defaults and bankruptcies (or near bankruptcies) associated
with the so-called subprime crisis by some of the largest financial and 
insurance companies illustrates that shocks
in the network of inter-dependencies of these companies can be sufficiently 
strong to destabilize them. Beyond these trivial examples, there is  
a large empirical literature on firm entries and exits, that suggests the need for
taking into account the existence of failure of large firms \cite{Dunne1,Dunne2,Bartelsmanetal03}.
To the extent that the empirical literature documents
a sizable exit at all size categories, we suggest that it is timely to
study a model with both firm exit at a size lower bound and due to a size-independent
hazard rate. Such a model constitutes a better approximation to the empirical data 
than a model with only firm exit at the lower bound.
\citeasnoun{Gabaix99} briefly considers an analogous situation (at least from a
formal mathematical perspective) and suggests that it may have an important impact on
the shape of the distribution of firm sizes.

To sum up, we consider an economy of firms undergoing continuous 
stochastic growth processes with births and deaths playing a central role at time scales as short as a few years. 
We argue that death processes are especially important to
understand the economic foundation of Zipf's law and its robustness.
In order to make our model closer to the data, we 
consider two different mechanisms for the exit of a firm: (\i) when the firm's 
size becomes smaller than a given minimum threshold and (\i\i) when an exogenous
shock occurs, modeling for instance operational risks, independently of the size
of the firm. The other important issue is to describe adequately the birth
process of firms. As a counterpart to the continuously active death process, we
will consider that firms appear according to a stochastic flow process
that may depend on macro-economic variables and other factors. The
assumptions underpinning this model as well as the main results derived from it 
are presented in section 2. Section 3 puts them in perspective in the light of 
recent theoretical models and empirical findings on 
the existence of deviations from Zipf's law. Section 4 provides 
complementary results which are important from an empirical point of view. All the proofs are gathered in the appendix at the end of the article.

\section{Exposition of the model and main results}

\subsection{Model setup}

We consider a reduced form model, with a first set of three assumptions, in which 
firms are created at random times $t_i$'s with initial random asset values 
$s_0^i$'s drawn from some given statistical distribution. More precisely:
\begin{assumption}
\label{assumption1}
There is a flow of firm entry, with births of new firms following a Poisson process
with exponentially varying intensity $\nu(t) =\nu_0 \cdot e^{d \cdot t}$, with $d \in {\mathbb R}$;
\end{assumption}

This assumption generalizes most previous approaches that address the question of modeling the size distribution of firms. In the basic model of \citeasnoun{Gabaix99} or in \citeasnoun{RHW2007a,RHW2007b}, all firms (or cities) are supposed to enter at the same time, which is technically equivalent to consider that there is only one firm in the economy. In Simon's models and in \citeasnoun{Luttmer07}, a flow of firms birth is considered, but births occur deterministically at discrete time steps (Simon) or continuously in time (Luttmer). Assumption~\ref{assumption1} allows for a {\em random} flow of birth.

As will be clear later on, the value of the parameter $\nu_0$ is not really relevant for the understanding of the shape of the distribution of firm sizes. In contrast, the parameter $d$, which characterizes the growth or the decline of the intensity of firm births, plays a key role insofar as it is directly related to the net growth rate of the population of firms.

We also assume that the entry size of a new incumbent firm is random, with a typical size which is time 
varying in order to account for changing installment costs, for instance. The size of a firm can represent its assets value, but for most of the developments in this article, the size could be measured as well by the number of employees or the sales revenues.
\begin{assumption}
\label{assumption:iii-a}
At time $t_i$, $i \in \mathbb N$, the initial size of the new entrant firm $i$ is given by $s_0^i=s_{0,i} \cdot e^{c_0 t_i}$, $c_0 \in \mathbb R$. The random sequence $\left\{s_{0,i}\right\}_{i \in \mathbb N}$ is the result of independent and identically distributed random draws from a common random variable $\tilde s_0$. All the draws are independent of the entry dates of the firms.
\end{assumption}

This assumption exogenously sets the size of entrant firms. It departs from \citeasnoun{Gabaix99} generalized model and \citeasnoun{Luttmer07} model by considering a
distribution of initial firm sizes that is unrelated to the distribution of
already existing firms. Besides, it does not imposes that all the firms enter with the same (minimum) size, as in \citeasnoun{Simon60} or \citeasnoun{Steindl65} which are retrieved by choosing a degenerated distribution of entrant firms and $c_0=0$. As we shall see later on, apart from the growth rate $c_0$ of the typical size of a new entrant firm, the characteristics of the distribution of initial firm sizes is, to a large extent, irrelevant for the shape of the upper tail of the steady-state distribution of firm sizes.

\begin{remark}
As a consequence of assumptions~\ref{assumption1} and \ref{assumption:iii-a}, the average capital inflow per unit time -- i.e. the average amount of capital invested in the creation of new firms per unit time -- is
\bea
dI(t) &=& \nu(t) \E \[\tilde s_0\] e^{c_0 t}\, dt~ ,\\
&=& \nu_0 \E \[\tilde s_0\] e^{(d+c_0) t}\, dt~ , \label{eq:invest}
\eea
and $d+c_0$ appears as the average growth rate of investment in new firms.
\end{remark}

As usual, we also assume that
\begin{assumption}
\label{assumption:Gibrat}
Gibrat's rule holds.
\end{assumption}
Assumption \ref{assumption:Gibrat} means that, in the continuous time limit, the size $S_i(t)$ of the $i^{th}$ firm of the economy at time $t \ge t_i$, conditional on its initial size $s_0^i$, is solution to the stochastic differential equation
\be
\label{jhojgfwv}
dS_i(t) = S_i(t) \( \mu \, dt + \sigma \,  dW_i(t)\)~, \qquad t \ge t_i ~, \qquad
S_i(t_i) = s_0^i~.  
\ee
The drift $\mu$ of the process can be interpreted as the
rate of return or the ex-ante growth rate of the firm. Its volatility is $\sigma$ 
and $W_i(t)$ is a standard Wiener process.
Note that the drift $\mu$ and the volatility $\sigma$ are the same for all firms.

This assumption together with assumption~\ref{assumption1} extends \citename{Simon55}'s model by allowing the creation of new firms at random times,  as already mentioned, and more importantly decouples the growth process of existing firms from the process of creation of new firms. It thus makes the model more realistic.

Let us now consider two exit mechanisms, based on the 
following empirical facts.  Referring to \citeasnoun{BD04}, the yearly rate of death of Italian firms is, on average, equal to $5.7\%$ with a maximum of about $20\%$ for some specific industry branches. \citeasnoun{Knaup05} examined the business survival characteristics of all establishments that started in the United States in the late 1990s when the boom of much of that decade was not yet showing signs of weakness, and finds that, if 85\% of firms survive more than one year, only 45\% survive more than four years. \citeasnoun{Brixy07} analysed the factors that influence regional birth and survival rates of new firms for 74 West German regions over a 10-year period. They documented significant regional factors as well as variability in time: the 5-year survival rate fluctuates between 45\% and 51\% over the period from 1983 to 1992. \citeasnoun{Bartelsmanetal03} confirmed that a large number of firms enter and exit most markets every year in a group of ten OECD countries: data covering the first part of the 1990s show the firm turnover rate (entry plus exit rates) to be between 15 and 20 percents in
the business sector of most countries, i.e., a fifth of firms are either recent entrants, or will close down within the year.

First of all, we assume that firms disappear when their asset values become smaller than some pre-specified minimum level $s_{\min}$.
\begin{assumption}
\label{assumption:c1}
There exists a minimum firm size $ s_{\min}(t) = s_1 \cdot e^{c_1 \cdot t}$, that varies at the constant rate $c_1 \le c_0$, below which firms exit.
\end{assumption}
This idea has been considered in several models of firm growth (see e.g. \citeasnoun{deWit05} and references therein) and can be related to the existence of a minimum efficient size in the presence of fixed operating costs. Besides, as for the typical size of new entrant firms, we assume that the minimum size of incumbent firms grows at the constant rate $c_1 \ge 0$, so that $s_{\min}(t):= s_1 e^{c_1 \cdot t}$. But $c_1$ is a priori different from $c_0$. It is natural to require that the lower bound $\underline{s}_0$ of the distribution of $\tilde s_0$ be larger than $s_1$ and that $c_0 \ge c_1$ in order to ensure that no new firm enters the economy with an initial size smaller than the minimum firm size and then immediately disappears\footnote{In fact, it seems that the typical size of entrant firms is much smaller than the minimum efficient size \cite[and references therein]{AA2001}. It means that two exit levels should be considered; one for old enough firms and another one for young firms. For tractability of the calculations, we do not consider this situation.}. The condition $s_1 e^{c_1 \cdot t}  < \underline{s}_0 e^{c_0 \cdot t}$ implies that the economy started at a time $t_0$ larger than
\be \label{eq:birthdate}
t_* = \frac{1}{c_1 - c_0} \cdot \ln \(\frac{\underline{s}_0}{s_1}\) < 0~.
\ee
We could alternatively choose $\underline{s}_0 = s_1$ so that the economy starts at time $t=0$. Another approach, suggested for instance by \citeasnoun{Gabaix99},
considers that firms cannot decline below a minimum size and remain in
business at this size until they start growing up again. Here, we have not used this rather artificial
mechanism.

Secondly, we consider that firms may disappear abruptly as the result of an unexpected large event (operational risk, fraud,...), even if their sizes are still large. Indeed, while it has been established that a first-order characterization for firm death involves lower failure rates for larger firms \cite{Dunne1,Dunne2}, \citeasnoun{Bartelsmanetal03} also state that, for sufficiently old firms, there seems to be no difference in the firm failure rate across size categories. Consequently
\begin{assumption}
\label{assumption:5}
There is a random exit of firms with constant hazard rate $h \ge \max \{-d, 0\}$ which is independent of the size and age of the firm.
\end{assumption}
\begin{remark}
As will become clear later on, the constraint $h \ge \max \{-d, 0\}$ is only necessary to guaranty that the distribution of firm sizes is normalized in the small size limit if there is no minimum firm size. The case $d>0$ ensures that the population of firms grows at the long term rate $d$ while the case $d<0$ allows describing an industry branch that first expands, then reaches a maximum and eventually declines at the rate $d$. Such a situation is quite realistic, as illustrated by figure 2 in \citeasnoun{Sutton97} which depicts the number of firms in the U.S. tire industry. Notice, in passing, that the case $h<0$ is also sensible. It corresponds to the situation considered by \citeasnoun{Gabaix99} in his generalized model, where firms are allowed to enter with an initial size randomly drawn from the size distribution of incumbent firms.
\end{remark}

Under assumptions~\ref{assumption1} and \ref{assumption:5}, i.e. not considering for the time being the mechanism of exit
of firms at the minimum size,  the average number $N_t$ of operating firms satisfies
\be
\frac{d N_t}{dt} + h N_t = \nu(t),
\ee
so that, assuming that the economy starts at $t=0$ for simplicity, we obtain
\be
N_t = \frac{\nu_0}{d+h} \[e^{d \cdot t} -  e^{-h \cdot t}\].
\ee
Consequently, the rate of firm birth, given by $\nu(t) /N_t$, is given by $\frac{d+h}{1 - e^{-(d+h) \cdot t}} \to d+h$ for $t$ large enough. The range of values of $d+h$ has been reported in many empirical studies. For instance, \citeasnoun{Reynolds94} give the regional average firm birth rates (annual firm births per 100 firms) of several advanced countries in different time periods: $10.4\%$ (France; 1981-1991), $8.6\%$ (Germany; 1986), $9.3\%$ (Italy; 1987-1991), $14.3\%$ (United Kingdom; 1980-1990), $15.7\%$ (Sweden; 1985-1990), $6.9\%$ (United States; 1986-1988). They also document a large variability from one industrial sector to another. More interestingly, \citeasnoun{BD04} as well as \citeasnoun{Dunne1} reports both the entry and exit rate for different sectors in Italy and in the US respectively. In every cases, even if sectorial differences are reported, the average aggregated entry and exit rates are remarquably close. This suggests that $d$ should be close to zero while $h$ is about $4-6\%$. The net growth rate of the population of firms, given by $\frac{1}{N_t} \frac{d N_t}{d t} = \frac{\nu(t)}{N_t} - h$ tends to $d$ for $t$ large enough, as announced after assumption~\ref{assumption1}. 

\subsection{Results}

Equipped with this set of five assumptions, we can now define
\be
\label{eq:m}
m:= \frac{1}{2} \[\(1 - 2 \cdot \frac{\mu - c_0}{\sigma^2}\) + 
\sqrt{\(1 - 2 \cdot \frac{\mu - c_0}{\sigma^2}\)^2 + 8 \cdot \frac{d+h}{\sigma^2}} \],
\ee
and derive our main result (see appendix~\ref{app:prop1} for the proof):
\begin{proposition}
\label{prop1}
Under the assumptions 1-5, provided that $\E\[\tilde s_0^m\] < \infty$, \\
for $t - t_* \gg \[ \(\mu - \frac{\sigma^2}{2} - c_0 \)^2 + 2 \sigma^2 (d+h)\]^{-1/2}$, the average distribution of firm's sizes follows an asymptotic power law with tail index $m$ given by (\ref{eq:m}),  in the following sense: the average number of firms with size larger than $s$ is proportional to $s ^{-m}$ as $s \to \infty$.
\end{proposition}

\begin{remark}
Condition $\E\[\tilde s_0^m\] < \infty$ in Assumption~\ref{assumption:iii-a} means that the fatness of the initial 
distribution of firm sizes at birth is less than the natural fatness resulting 
from the random growth. Such an assumption is not always satisfied, in 
particular in \citeasnoun{Luttmer07}'s model where, due to imperfect imitation, 
the size of entrant firms is a fraction of the size of incumbent firms.
\end{remark}

One can see that the tail index increases, and therefore the distribution of firm sizes becomes thinner tailed, as 
$\mu$ decreases and as $h$, $c_0$, and $d$ increase. This dependence can be easily rationalized. Indeed, the smaller the expected growth rate $\mu$, the smaller the fraction of large firms, hence the thinner the tail of the size distribution and the larger the tail index $m$. The larger $h$, the smaller the probability for a firm to become large, hence a thinner tail and a larger $m$. As for the impact of $c_0$, rescaling the firm sizes by $e^{c_0 \cdot t}$, so that the mean size of entrant firms remains constant, does not change the nature of the problem. The random growth of firms is then observed in the moving frame in which the size of entrant firms remains constant on average. Therefore, the size distribution of firms is left unchanged up to the scale factor $e^{c_0 \cdot t}$. Since the average growth rate of firms in the new frame becomes $\mu'=\mu - c_0$, the larger $c_0$, the smaller $\mu'$, hence the smaller the probability for a firm to become 
relatively larger than the others, the thinner the tail of the distribution of firm sizes and thus the larger $m$. Finally, the larger $d$ is, the larger the fraction of young firms, which leads to a relatively larger fraction of firms with sizes of the order of the typical size of entrant firms
and thus the upper tail  of the size distribution becomes relatively thinner and $m$ larger.

As a natural consequence of
proposition~\ref{prop1}, we can assert that
\begin{corollary}
\label{corol1}
Under the assumptions of proposition~\ref{prop1}, the mean distribution of firm sizes admits a
well-defined steady-state distribution which follows Zipf's law (i.e. $m=1$) if,
and only if, 
\be
\mu - h  = d + c_0~.
\label{ghtwbvgw}
\ee
\end{corollary}
\begin{remark}
In an economy where the amount of capital invested in the creation of new firms is constant per unit time, namely
\be
\label{eq:mhqeg}
\nu(t) \cdot s_0(t) = const.~ ,
\ee
we necessarily get $d+ c_0 = 0$ so that the balance condition reads $\mu = h$.
\end{remark}

To get an intuitive meaning of the condition in corollary~\ref{corol1}, let us state the following result (see the proof in appendix~\ref{app2}):
\begin{proposition}
\label{prop2}
Under the assumptions of proposition~\ref{prop1}, the long term average growth rate of  the overall economy is $\max \left\{\mu-h, d+c_0 \right\}$.
\end{proposition}

The term $d + c_0$ quantifies the growth rate of investments in new entrant firms, resulting
from the growth of the number of entrant firms (at the rate $d$) and the growth of the size of new entrant firms (at the rate $c_0$). The term $d$ reflects several factors, including improving pro-business legislation and tax laws as well as increasing entrepreneurial spirit. The latter term $c_0$ is essentially due to time varying installment costs, which can be negative in a pro-business economy. 

The other term $\mu - h$ represents the average growth rate of an incumbent firm. Indeed, considering a running firm at time $t$, during the next instant $dt$, it will either exit with probability $h \cdot dt$ (and therefore its size
declines by a factor $-100\%$) or grow at an average rate equal to $\mu \cdot
dt$, with probability $(1- h \cdot dt)$. The coefficient $\mu$ can be called the conditional growth rate
of firms, conditioned on not having died yet. Then, the expected growth rate over the 
small time increment $dt$ of an incumbent firm is $(\mu-h) \cdot dt + O\(dt^2\)$. As shown by the following equation, drawn from appendix~\ref{app2}, the average size of the economy $\Omega(t)$ (if we neglect the exit of firms by lack of a sufficient size) reads
\be
\Omega(t) = \int_0^t e^ {(\mu-h) \cdot (t-u)} dI(u)~,
\ee
where $I(t)$ is the average capital inflow invested in the creation of news firms per unit time (see eq. \ref{eq:invest}). Thus $\mu-h$ is also the return on investment of the economy.

Thus, the long term average growth of the economy is driven either by the growth of investments in new firms, whenever $d+c_0 > \mu - h$, or by the growth of incumbent firms, whenever $\mu - h > d+ c_0$. The former case does not really make sense, on the long run. Indeed, it would mean that the growth of investments in new firms can be sustainably larger than the rate of return of the economy. Such a situation can only occur if we assume that the economy is fueled by an inexhaustible source of capital, which is obviously unrealistic. 
As a consequence, it is safe to assume $\mu-h \ge d+ c_0$ on the long run. The regime 
$d+c_0 > \mu - h$ might however describe transient bubble regimes developing under
unsustainably large capital creation \cite{Brookings}.

\begin{proposition}
\label{prop:main}
In a growing economy whose growth is driven by that of incumbent firms, the tail index of the size distribution is such that $m \le 1$.\\
Along a balanced growth path, which corresponds to a maximum sustainable growth rate of the investment in new firms, the tail index of the size distribution is equal to one.
\end{proposition}

Proposition~\ref{prop:main} shows that Zipf's law characterizes an efficient and sustainable allocation of resources among the firms of an economy. Any deviation from it is the signature of an inefficiency and/or an unsustainability of the allocation scheme. In this respect, the size distribution of firms is a diagnostic device to assess the efficiency and the sustainability of the allocation of resources among firms in an economy.

\begin{proof} According to the natural assumption that the growth of the economy is driven by the growth of incumbent firms, i.e. $\mu - h \ge d+c_0$, we get $d+h \le \mu - c_0$ and $\epsilon \le 1$ which leads to $m \le 1$ (see illustration on figure~\ref{FigEpsilon}); we have used assumption~\ref{assumption:5} according to which $d+h \ge 0$ hence $\mu - c_0 \ge 0$. On a balanced growth path, both investments in new firms and incumbent firms grow at the same rate $\mu-h = d + c_0$, hence the growth rate of the investment in new firms is maximum and by corollary~\ref{corol1} the tail index $m$ of the size distribution equals one.
\end{proof}
\begin{figure}
\centerline{\includegraphics[width=0.75\textwidth]{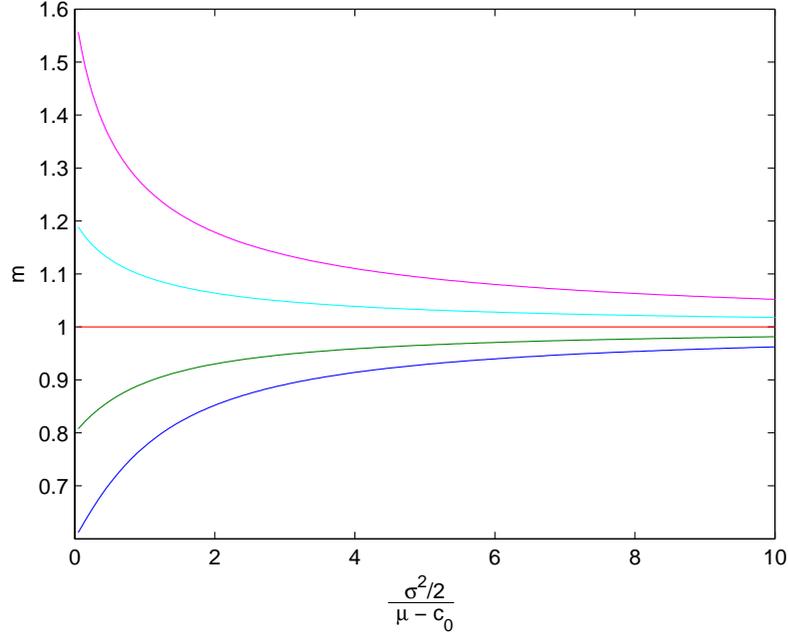}}
\caption{\label{FigEpsilon} The figure shows the exponent $m$ of the power law tail of the
distribution of firm sizes, given by (\ref{eq:m}), as
a function of $\frac{\sigma^2/2}{\mu-c_0}$, for different values of the ratio $\epsilon:=\frac{d+h}{\mu - c_0}$. Bottom to top $\epsilon = 0.6; 0.8; 1; 1.2; 1.6$.}
\end{figure}

In the present framework, the crucial parameters  $d, c_0, \mu$ and $h$ are exogenous. 
While this is beyond the scope of the present paper, we can however surmise that, within
an endogenous theory in which the growth of investments would be naturally correlated with the growth
of the firms in the economy because the success of firms generates the cash flow at the source
of new investments, the balance growth condition (\ref{ghtwbvgw}) appears 
almost unavoidable for a sustainable development. It is quite remarkable that
Zipf's law derives as the robust statistical translation of this balance growth condition.

\begin{remark} \label{rem:4}
Our theory suggests two simple explanations for the empirical evidence that the exponent $m$ is close to $1$. Either the investment in new firms is close to its maximum sustainable level so that the balance condition is approximately satisfied, or the volatility $\sigma$ of incumbent firms sizes is large. Indeed, according to equation~\eqref{eq:m}, the tail index $m$ goes to one as $\sigma$ goes to infinity irrespective of the values of the parameters $d, \mu, c_0$ and $h$. In fact, the larger the volatility, the larger the tolerance to the departure from the balance condition. Indeed, expanding relation \eqref{eq:m} for $\sigma$ large, we get
\be
m=1 - 2 \cdot \frac{\mu - h - c_0 - d}{\sigma^2} + 4 \cdot \frac{(d+h) \(\mu - h - c_0 - d\)}{\sigma^4} + O\(\frac{1}{\sigma^{6}}\),
\ee
and for small departures from the balance condition
\be
m=1 - \frac{2}{1+2\frac{d+h}{\sigma^2}} \cdot \frac{\mu -c_0 -h -d}{\sigma^2} + \frac{8 \frac{d+h}{\sigma^2}}{\(1+2\frac{d+h}{\sigma^2}\)^3} \cdot \(\frac{\mu -c_0 -h -d}{\sigma^2}\)^2+ O\(\(\mu -c_0 -h -d\)^3\).
\ee

When the volatility changes, the convergence of the size distribution toward its long-term distribution may be faster or slower. Indeed, according to Proposition~\ref{prop1}, the size distribution converges to a power law when the age of the economy is large compared with\linebreak $ \[ \(\mu - \frac{\sigma^2}{2} - c_0 \)^2 + 2 \sigma^2 (d+h)\]^{-1/2}$. This quantity is a decreasing function of the volatility if (and only if) $\frac{\sigma^2}{2} > \( \mu - c_0\) - 2 \(d+h \)$. Therefore, when the volatility is large Zipf's becomes more robust {\em and} the convergence towards Zipf's law is faster.
\end{remark}

\begin{remark} \label{rem:infmean}
The regime where $m \leq 1$,  which predicts an infinite mean
size, seems to violate the constraint that there is  a finite amount of capital
(or, employees) in the economy. This suggests that the associated parameter ranges
are just not possible in actual economies. Actually, the regime $m \leq 1$ is perfectly 
possible, as least in an intermediate asymptotic regime. Indeed, a real economy which grows
at a non-vanishing growth rate bounded by zero from below is finite only because
it has a finite age. As explained  in section \ref{sec:efi}, the distribution of firm sizes in such
a finitely lived economy (arguably representing the real world) is characterized by 
a power law regime with exponent $m$ as given by Proposition 1, crossing over to a
faster decay at very large firm sizes. The cross-over regime occurs for larger and larger
firm sizes as the age of the economy increases. There is thus no contradiction between
the finiteness of the amount of capital in the economy and the power law with exponent
$m<1$ up to an upper domain, so that the mean does exist. In other words, 
the paradox is resolved by correctly ordering the two limits: (i) limit of larger firm sizes ${\rm lim}_{S \to +\infty}$;
(ii) limit of large age of the economy ${\rm lim}_{\theta \to +\infty}$. The correct ordering for 
a finite and long-lived economy is ${\rm lim}_{\theta \to +\infty} {\rm lim}_{S \to +\infty}$, which means 
that, taking the limit of large firm sizes at fixed large but finite age $\theta$ leads to a finite mean,
coexisting with a power law intermediate asymptotic with exponent $m$ given by Proposition 1.
\end{remark}

\subsection{Calibration to empirical data}

According to \citeasnoun[table 2]{Dunne1}, the relative size of entrant firms to incumbent firms seems to have slightly declined during the period 1963-1982 in the US. According to our model, the ratio of the average size of entrant firms to the average size of incumbent firms is, for large enough time $t$,
\be
\label{eq:elku}
\frac{s_0 \cdot e^{c_0 \cdot t}}{ \Omega_t / N_t} \sim 
\begin{cases}
\frac{\mu-h-d-c_0}{d+h} \cdot e^{-\(\mu-h-d-c_0\) \cdot t},& \qquad \text{provided that~} \mu-h > d+c_0,\quad \quad(a) \\
\frac{1}{d+h} \cdot \frac{1}{t},& \qquad \text{provided that~} \mu-h = d+c_0, \quad \quad(b)\\
\frac{d+h-\mu+c_0}{d+h},& \qquad \text{provided that~} \mu-h < d+c_0,\quad \quad(c)\\
\end{cases}
\ee
where $\Omega_t$ is the average size of all incumbent firms (see Appendix~\ref{app2}) and $N_t$ is the average number of incumbent firms, at time $t$. The fact that \citeasnoun{Dunne1} observe a slight decay in the relative size of entrant firms to incumbent firms suggests that the condition of sustainable growth $\mu-h > d+c_0$ holds. Under this hypothesis, the calibration of equation (\ref{eq:elku}.a) by OLS gives, on an annual basis, 
\be
\mu-h-d-c_0 = 1.8\%~ (1.2\%) \qquad \text{and} \qquad \frac{\mu-h-d-c_0}{d+h} = 28\%~ (3\%).
\ee
The figures within parenthesis provide the standard deviations of the estimates. As a consequence, the alternative
hypothesis $\mu-h-d-c_0 \leq 0$ cannot be rejected at any usual significance level and we cannot affirm
that  Dunne's data corresponds to the regime $\mu-h-d-c_0 > 0$.

Under the second hypothesis $\mu-h=d-c_0$, equation (\ref{eq:elku}.b) leads to test the null hypothesis that the slope of the OLS regression of the logarithm of the size of entrant firms relative to the size of incumbent firms against the logarithm of time is equal to $-1$. Instead, we estimate a slope equal to $-0.086$	$(0.068)$, which is therefore not significantly different from zero. Thus, we reject the hypothesis $\mu-h=d-c_0$.

According to equation (\ref{eq:elku}.c), the size of entrant firms relative to the size of incumbent firms is constant over the period under consideration.
To formally test this hypothesis, we perform the OLS regression of the size of entrant firms versus to the size of incumbent firms against time. We find that the hypothesis of a time dependent ratio of the size of entrant firms relative to the size of incumbent firms is rejected at any usual significance level. We thus have to conclude that the third alternative actually holds and we get
\be
\frac{d+h-\mu+c_0}{d+h} = 25\%.
\ee
With the figures $d=0$ and $h=5\%$ obtained from \citeasnoun{Dunne1}, we obtain ${d+h-\mu+c_0} = 1.25\%$ and $\mu-c_0 = 3.75\%$. Thus, the balance condition is not strictly satisfied but the observed departure from the balance condition remains weak.

To sum up, reasonable estimates of the key parameters are $h = 4-6\%$, $d=\pm 0.5\%$, $\mu-c_0 = h \pm 2\%$. As for $\sigma$, \citeasnoun{Buldyrev97} report the standard deviations of the growth rates in terms of sales, assets, cost of goods sold and plant property and equipment for US publicly-traded companies. \citeasnoun{Buldyrev97} find that $\sigma$ ranges typically between $30\%$ to $50\%$. Based upon this set of figures, relation~\eqref{eq:m} leads to a tail index $m$ ranging between $0.7$ and $1.3$, in agreement with the range of values usually reported in the literature.

Proposition~\ref{prop1} states that the asymptotic power law of the distribution of firm sizes can be observed if the age of the economy is large compared with $\[ \(\mu - \frac{\sigma^2}{2} - c_0 \)^2 + 2 \sigma^2 (d+h)\]^{-1/2}$. With the set of parameters above, this corresponds to economies whose age is large compared to $5$ to $12$ years.

\section{Discussion}

\subsection{Comparison with Gabaix's model}

Corollary~\ref{corol1} seems reminiscent of the condition 
given by \citeasnoun{Gabaix99} in its basic model, which relies on the argument that, because they are all born at the same time, firms grow -- on average -- at the same rate as the overall economy. 
Consequently, when discounted by the global growth rate of the
economy, the average expected growth rate of the firms must be zero.
Applied to our framework, and focusing on the distribution of {\em discounted} firm sizes, 
this argument would lead to $\mu=h$, with $d=c_0=c_1=0$ in order to match Gabaix's assumptions.
\citeasnoun{Gabaix99}'s condition would thus seem to be equivalent to our
balance condition for Zipf's law describing the density of firms' sizes to hold. 

Actually, this reasoning is incorrect. Consider the case where $\mu>h$, such that the global economy grows at the average growth rate $r_{G}=\mu-h$ according to Proposition~\ref{prop2}. \citeasnoun{Gabaix99} proposed to measure the growth of a firm in the frame of the global economy. In this moving frame, the conditional average growth rate of the firm is $\mu'=\mu-r_{G}=h$, which indeed would suggest that the balance condition is {\it automatically} obeyed when $\mu$ is replaced by $\mu'$. But, one should notice that $\mu'$ is a transformed growth rate, and not the true rate. The average growth rate $r_{G}=\mu-h$ of the global economy is micro-founded on the contributions of all growing firms. 
It would be incorrect to insert $\mu'$ in the statements of Proposition 1, as $\mu'$ is the effective growth rate resulting from the change of frame, while our exact derivation requires the 
parameters $\mu$ and $h$ for Proposition 1 to hold.  As such, nothing in our model automatically
sets the growth rate $\mu$ of firms to their death rate $h$, contrarily to what happens in \citeasnoun{Gabaix99}'s
model. The main difference that invalidates the application of \citeasnoun{Gabaix99}'s argument is the stochastic flow of firm's births and deaths.

It is important to understand that in \citeasnoun{Gabaix99}'s basic model, the derivation of Zipf's law relies crucially on a model view of the economy in which {\em all firms are born at the same instant}. Our approach is thus
essentially different since it considers the flow of firm births, as well as their deaths, which is more in agreement with empirical evidence. Note also that  the available empirical evidence on Zipf's law is based on analyzing {\it cross-sectional} distributions of firm sizes, i.e., at specific times. As a consequence, the change to the global economic growth frame, argued by \citeasnoun{Gabaix99}, just amounts to multiplying  the value of each firm by the same constant of normalization, equal to the size of the economy at the time when the cross-section is measured. Obviously, this normalization does not change the exponent of the power law distribution of sizes, if it exists. Furthermore, elaborating on \citeasnoun{Krugman96}'s argument about the non-convergence of the distribution of firm sizes toward Zipf's law in \citeasnoun{Simon55}'s model, \citeasnoun{Blank_Solomon00} have shown that \citeasnoun{Gabaix99}'s argument suffers from a more technical problem. Based on the demonstration that the two limits, the number of firms $N \to \infty$ and $s_{\rm min}(t)/\Omega(t) \to 0$\footnote{The term $\Omega(t)$ refers to the average size of the economy defined by the sum of the sizes over the population of incumbent firms (see \eqref{eq:lkesu} in appendix~\ref{app2}).} (or equivalently the limit of large times $t \to \infty$) are non-commutative, \citeasnoun{Blank_Solomon00} showed that Zipf's exponent $m=1$ as obtained by \citeasnoun{Gabaix99}'s argument requires (i) taking the long time limit $s_{\rm min}(t)/\Omega(t) \to 0$ over which the economy made of a large but finite number $N$ firms grows without bounds, while simultaneously obeying the condition (ii) $N \gg \exp[\Omega(t) / s_{\rm min}(t)]$. The problem is that conditions (i) and (ii) are mutually exclusive. \citeasnoun{Blank_Solomon00} showed that this inconsistency can be resolved by allowing the number of firms to grow proportionally to the total size of the economy.

In a generalized approach of his basic model, Gabaix accounts for the appearance of new entities with a constant rate $\nu$  (equal to $d+h$ with our notations) and shows \cite[Proposition 3]{Gabaix99} that, as long as this birth rate is less than the growth rate  $\gamma$ of existing entities ($\mu$ or $\mu-h$ in our notations), the results of his basic model holds, i.e., Zipf's law holds. On the contrary, he shows that the tail index of the size distribution is equal to $m$ given by \eqref{eq:m} when the birth rate of new entities is larger than their growth rate. This result seems in contradiction with ours, as well as with \citeasnoun{Luttmer07}'s results, insofar as Proposition~\ref{prop1} states that $m$ is the tail index of the size distribution irrespective of the relative magnitude of the birth rate of entrant firms and of the growth rate of incumbent ones. The discrepancy between these two results comes from an error in  Gabaix's proof of Zipf's law in the regime when the birth rate of new entities is less than the growth rate of existing entities. The error consists in assuming that young firms do not contribute at all to the shape of the tail of the size distribution when $\nu$ is less than $\gamma$ 
\footnote{
To show that Zipf's law holds as long as the birth rate of new entities is less than the growth rate of existing entities, \citeasnoun[Appendix 2]{Gabaix99} splits the population of cities in two parts: the old ones, whose age is larger than $T=t/2$, and the young ones, whose age is smaller than $T=t/2$. In the limit of large time $t$, he shows that the size distribution of old cities should follow Zipf's law as a consequence of the results derived from his basic model. Then he provides the following majoration of the size distribution of firms born at time $\tau > t/2$, i.e., for young firms: $\Pr \[S > s | {\rm birthdate}= \tau\] \le \E[S |{\rm birthdate}=\tau]/s$. This trivial inequality requires the expectation $\E[S |{\rm birthdate}=\tau]$ be finite. Thus, Gabaix's derivation crucially relies on the fact that the firms whose ages are slightly larger than $t/2$ are old enough for Zip's law to hold (and thus for 
$\E[S |{\rm birthdate} < t/2]$ to be arbitrary large and infinite for an arbitrarily large economy which allows for the sampling
of the full distribution), while the firms whose ages are slightly less than $t/2$ are not old enough for Zipf's law to hold, and 
therefore they still admit a finite average size: $\E[S |{\rm birthdate} > t/2]< \infty$. It is clearly a contradiction 
as one cannot have simultaneously $\E[S |{\rm birthdate} = (t/2)^+] < \infty$ and Zipf's law for times $(t/2)^-$.
This invalidates eq. (17) in \citeasnoun{Gabaix99} because the integrand have to diverge as $\tau \to T$, with the notations of Gabaix article.}.
%We shows in Appendix~\ref{app4} that, on the contrary, the shape of the distribution of firms size is always controlled by young firms even in the absence of firm failure rate ($h=0$).
Therefore, in the presence of firm entries, Gabaix's approach does not allow to explain Zipf's law.

\subsection{Comparison with Luttmer's model}

Based upon structural models, an important modeling strategy has been developed, starting from \citeasnoun{Lucas1978} and evolving to the more recent \citeasnoun{Luttmer07,Luttmer08} or \citeasnoun{RHW2007a,RHW2007b} models. The distribution of firm sizes then appears as one of the properties of a general equilibrium model, which depends on different industry parameters. In these models, Zipf's law is obtained as a limit case, needing a rather sharp fine tuning of the control parameters. Rossi-Hansberg and Wright's model is a ``one firm'' model as in \citeasnoun{Gabaix99} and is therefore subjected to the same restrictions. We do not discuss further this model in light of the results of our reduced form model. 
In contrast, the assumptions underpinning  Luttmer's model match the assumptions under which proposition~\ref{prop1} holds, which motivates a closer comparison.

\citeasnoun{Luttmer07} considers an economy of firms with different ages.
For a firm of age $a$, its size  $S_a$ follows a geometric Brownian motion
\be
d \ln S_a = \mu \cdot d a + \sigma  \cdot d W_a
\ee
where the drift $\mu$ and volatility $\sigma$ are derived from a micro-economic model and are related to the price elasticity $\frac{\beta}{1-\beta}$ of the demand for commodity, to the rate $\theta_E$ at which the productivity of entering firms grows over time, to the trend $\theta_I$ of log productivity for incumbent firms, and to the volatility  $\sigma_Z$ of the productivity:
\be
\mu = \frac{\beta}{1 - \beta} \(\theta_I - \theta_E \), \qquad \sigma = \frac{\beta}{1 - \beta} \sigma_Z.
\ee
Due to the presence of fixed costs, incumbent firms exit when their size reaches a constant minimum size $b$ and in this case only. In our notations, this implies $c_1=0$ and $h=0$. In addition, Luttmer assumes that the overall number of incumbent firms grows at a rate $\eta > \mu + \sigma^2/2$ so that the size of a typical incumbent firm is not expected to grow faster than the population growth rate. Within our framework, the number of firms grows, on the long run, at the rate $d$, so that we have the correspondence $d=\eta$. Finally, Luttmer considers that firms enter either with a fixed size or with a size taken from the same distribution as the incumbent firms; consequently, in our notations, we have $c_0=0$. Then, by application of proposition~\ref{prop1}, we conclude, as in \citeasnoun[section III.B]{Luttmer07}, that the size distribution of firms follows a power law with a tail index given by
\be
m = - \frac{\mu}{\sigma^2} + \sqrt{\(\frac{\mu}{\sigma^2}\)^2 + 2 \frac{\eta}{\sigma^2}}~.
\ee
Notice that Luttmer only considers the long term distribution of firm sizes, while our result allows considering the transient regime which eventually leads to the power law. In particular, accounting for the transient regime avoids resorting to the assumption $\eta > \mu + \sigma^2/2$. Indeed, in Luttmer's model, this assumption ensures that the tail index $m$ remains larger than one so that, detrended by the overall growth of the number of firms given by $e^{+\eta t}$, the average firm size is finite. This is a natural requirement if the economy is assumed to be finite. This latter assumption is more questionable in economies of infinite duration. In contrast, when finite time effects are considered as in our framework, the average firm size is always finite for finite times, since the density of firm sizes decays faster than any power law beyond the 
intermediate asymptotic described by the power law, whether $\eta > \mu + \sigma^2/2$ or not. In other words,
the power law is truncated by a finite time effect, as derived in appendix~\ref{app:prop1}. As time increases, the
truncation recedes progressively to infinity, thus enlarging the domain of validity of the power law. The 
exact power law distribution is attained therefore for asymptotically large times (we come back to this point latter on in section~\ref{sec:efi}). Therefore, the constraint $\eta > \mu + \sigma^2/2$ is not necessary anymore
in our framework, since there is no reason for the average firm size to remain finite at infinite times when the size of the overall economy becomes itself infinite. 

The endogeneization of the growth rate of the productivity of entrant firms performed by Luttmer in the second part of his article does not match our assumptions, so that we cannot proceed further with the comparison of his results with ours. Indeed, in Luttmer's case, the upper tail of the distribution of entrant firms behaves as the tail of the distribution of incumbent firms, so that assumption~\ref{assumption:iii-a} is not satisfied.

\section{Miscellaneous results}

\subsection{Distribution of firms' age and declining hazard rate}

\citeasnoun{BPZ92}, \citeasnoun[and references therein]{Caves98} or \citeasnoun{Dunne1,Dunne2}, among others, have reported declining hazard rates with age. Under assumption~\ref{assumption:5}, the hazard rate is constant, which seems to be counterfactual. However, we now show that the presence of the lower barrier below which firms exit allows to account for age-dependent hazard rate. 

Let us denote by $\theta$ the age of a firm at time $t$, i.e., the firm was born at time $t-\theta$. Expression \eqref{eq:sdmmosit} in appendix~\ref{app:prop1} allows us to derive the probability that, at time $t$, a firm older that $\theta$ is still alive, which corresponds to the distribution of firm ages. Indeed denoting by $\tilde \Theta_t$ the random age of the considered firm at time $t$,
\be
\Pr \[ \tilde \Theta_t > \theta\] = \displaystyle \int_{s_{\min}(t)}^\infty
\frac{1}{s} \varphi\left[\ln\left(\frac{s}{s_{\min}(t)}\right);t,\theta\right] \, ds ,
\ee
where $\frac{1}{s} \varphi\left[\ln\left(\frac{s}{s_{\min}(t)}\right);t,\theta\right]$ is the size density of firms of age $\theta$ at time $t$. Some algebraic manipulations give
\bea
\Pr \[ \tilde \Theta_t > \theta\] &=& \displaystyle  \frac{1}{2} \[ {\rm erfc} \left(-\frac{\ln\rho(t) + (\delta -1-\delta_0)\tau}{ 2 \sqrt{\tau}} \right) \right. \\
&-& \left. \rho(t)^{1 - \delta + \delta_0} \cdot {\rm erfc} \left(\frac{\ln\rho(t) - (\delta -1-\delta_0)\tau}{ 2 \sqrt{\tau}} \right) \]~,
\eea
with $\tau:=\frac{\sigma^2}{2} \theta$, $\delta:=\frac{2 \mu}{\sigma^2}$ and $\delta_0:=\frac{2 c_0}{\sigma^2}$.

Accounting for the independence of the random exit of a firm with hazard rate $h$ from the size process of the firm
(assumption~\ref{assumption:5}), the ``total'' hazard rate reads
\bea
{\mathcal H}(t,\theta) &=& h - \frac{d \ln \Pr \[ \tilde \Theta_t > \theta\]}{d \theta},\\
&=&  h + \frac{\ln \(\frac{s_0(t)}{s_{\min}(t)}\) \cdot \(\frac{s_0(t)}{s_{\min}(t)}\)^{- \frac{1 - \delta + \delta_0}{2}} \cdot \exp \[- \frac{\ln^2 \(\frac{s_0(t)}{s_{\min}(t)}\) + \(1 - \delta + \delta_0\)^2 \tau^2}{4 \tau} \]}{{\rm erfc} \left(-\frac{\ln \frac{s_0(t)}{s_{\min}(t)} + (\delta -1-\delta_0)\tau}{ 2 \sqrt{\tau}} \right)  -  \(\frac{s_0(t)}{s_{\min}(t)}\)^{1 - \delta + \delta_0} \cdot {\rm erfc} \left(\frac{\ln \frac{s_0(t)}{s_{\min}(t)} - (\delta -1-\delta_0)\tau}{ 2 \sqrt{\tau}} \right)}, \label{theynbw}
\eea
assuming, for simplicity, that the random variable $\tilde s_0$ reduces to a degenerate random variable $s_0$. Expression (\ref{theynbw}) shows that the failure rate actually depends on firm's age. It also depends explicitly on the current time $t$ through the ratio $\frac{s_0(t)}{s_{\min}(t)}$.

Let us focus on the case $c_0=c_1$, which corresponds to the same growth rate for $s_0(t)$ and $s_1(t)$. This allows considering arbitrarily old firms since, according to \eqref{eq:birthdate}, the starting point of the economy can then be $t_* = - \infty$. We obtain the limit result 
\be
{\mathcal H}(t,\theta) \stackrel{\theta \to \infty}{\longrightarrow}
\begin{cases}
\displaystyle h,& \quad \mu -c_1 - \frac{\sigma^2}{2} > 0,\\
\displaystyle h+\frac{1}{2 \sigma^2}\(\mu -c_1 - \frac{\sigma^2}{2} \)^2,& \quad \mu -c_1 - \frac{\sigma^2}{2} \le  0~.
\end{cases}
\ee
In the moving frame of the exit barrier, $\mu -c_1 - \frac{\sigma^2}{2}$ is the drift of the log-size of a firm
\be
d \ln S(t) = \(\mu -c_1 - \frac{\sigma^2}{2}\) dt + \sigma dW(t).
\ee
Thus, when the drift is positive, the firm escapes from the exit barrier, i.e., its size grows almost surely to infinity, so that the firm can only exit as the consequence of the hazard rate $h$. On the contrary, when the drift is non-positive, the firm size decreases and reaches the exit barrier almost surely, so that the firm exits either because it reaches the exit barrier or because of the hazard rate $h$. Hence the result that the asymptotic total failure rate is the sum of the exogenous hazard rate $h$ and of the asymptotic endogenous hazard rate $\frac{1}{2 \sigma^2}\(\mu -c_1 - \frac{\sigma^2}{2} \)^2$ related to the failure of a firm when it reaches the minimum efficient size in the absence of $h$ \footnote{Mathematically speaking, this hazard rate can be derived form the generic formula that gives the probability that a Brownian motion $\left\{X_t \right\}_{t \ge 0}$ with negative drift, started from $X_0 >0$, crosses for the first time  the lower barrier $X=0$.}.

Differentiating the age-dependent hazard rate given by (\ref{theynbw})
with respect to $\theta$ and using the asymptotic expansion of the error function \cite{AB1965}, we get
\be
\partial_\theta {\cal H}(t, \theta)=
\begin{cases}
\displaystyle - \frac{1}{2 \sigma^2}\(\mu -c_1 - \frac{\sigma^2}{2} \)^2 \cdot {\cal H}(t, \theta) \cdot \[ 1 + O\(\frac{1}{\theta}\)\], \quad & \mu -c_1 - \frac{\sigma^2}{2} > 0,\\
\displaystyle - \frac{3 \sigma^2}{ \theta^2}\(\mu -c_1 - \frac{\sigma^2}{2} \)^{-2} {\cal H}(t, \theta) \cdot \[ 1 + O\(\frac{1}{\theta}\)\], \quad & \mu -c_1 - \frac{\sigma^2}{2} \le 0,\\
\end{cases}
\ee
which shows that the total failure rate decreases with age, at least for large enough age, in agreement with the literature.

\subsection{Deviations from Zipf's law due to the finite age of the economy}
\label{sec:efi}

Considering, for simplicity, that $\tilde s_0$ is a degenerate random variable such that $\Pr[\tilde s_0 = s_0]=1$, we can determine the deviations from the asymptotic power law tail of the mean density of firm sizes (given explicitly by \eqref{gstlimit} in appendix~\ref{app:prop1}) due to the finite age of the economy. For this, it is convenient to study the  $s$-dependence of the mean number of firms whose sizes exceeds a given level $s$:
\begin{equation}\label{Nintdef}
N(s,t) = \int_s^\infty g(s',t) ds'~ .
\end{equation}
Zipf's law corresponds to $N(s,t)\sim s^{-1}$ for large $s$.

All calculations done, defining $s_0(t):=s_0 e^{c_0 \cdot t}$ as being the initial size of an entrant firm at time $t$ when $\Pr[\tilde s_0 = s_0]=1$, we obtain the number $N(\kappa,\tau)$ of firms whose normalized size $\frac{s}{s_0(t)}$ is larger than $\kappa$ at the standardized age $\tau := \frac{\sigma^2}{2} \theta$,
\begin{equation}\label{numbsanothrep}
N(\kappa,\tau) = B_-\, \kappa^{-\varrho_-}+
B_+\, \kappa^{+\varrho_+} - C~ ,
\end{equation}
where
\begin{equation}\label{bmpdef}
\begin{array}{c} \displaystyle
B_- := {1 \over 2\alpha(\eta) \varrho_-} \Big[ \text{erfc}\left({\ln\kappa- \tau
\alpha(\eta) \over 2 \sqrt{\tau}} \right) - \(\frac{s_0(t)}{s_{\min}(t)}\)^{-\alpha(\eta)}
\text{erfc}\left({\ln(\kappa\rho^2)- \tau \alpha(\eta) \over 2 \sqrt{\tau}}
\right)\Big] ~ ,
\\[4mm] \displaystyle
B_+ := {1 \over 2\alpha(\eta) \varrho_+} \Big[ \text{erfc}\left({\ln\kappa+ \tau
\alpha(\eta) \over 2 \sqrt{\tau}} \right) - \(\frac{s_0(t)}{s_{\min}(t)}\)^{\alpha(\eta)}
\text{erfc}\left({\ln(\kappa\rho^2)+ \tau \alpha(\eta) \over 2 \sqrt{\tau}}
\right)\Big]~ ,
\\[4mm] \displaystyle
C:= {1 \over 2\eta} e^{-\eta \tau} \Big[\text{erfc}\left({\ln\kappa- \tau \alpha
\over 2 \sqrt{\tau}} \right)- \(\frac{s_0(t)}{s_{\min}(t)}\)^{-\alpha}
\text{erfc}\left({\ln(\kappa\rho^2)- \tau \alpha \over 2 \sqrt{\tau}} \right)
\Big] ~ ,
\end{array}
\end{equation}
and $\varrho_\pm := \frac{1}{2} \[ \alpha \pm \alpha\(\eta\)\]$, $\alpha:= 2\cdot \frac{\mu - c_0}{\sigma^2}- 1$, $\alpha(\eta) := \sqrt{\alpha^2 + 4 \eta}$, $\eta :=\frac{\sigma^2}{2}\,(d+h)$.

\begin{figure}
\centerline{\includegraphics[width=0.75\textwidth]{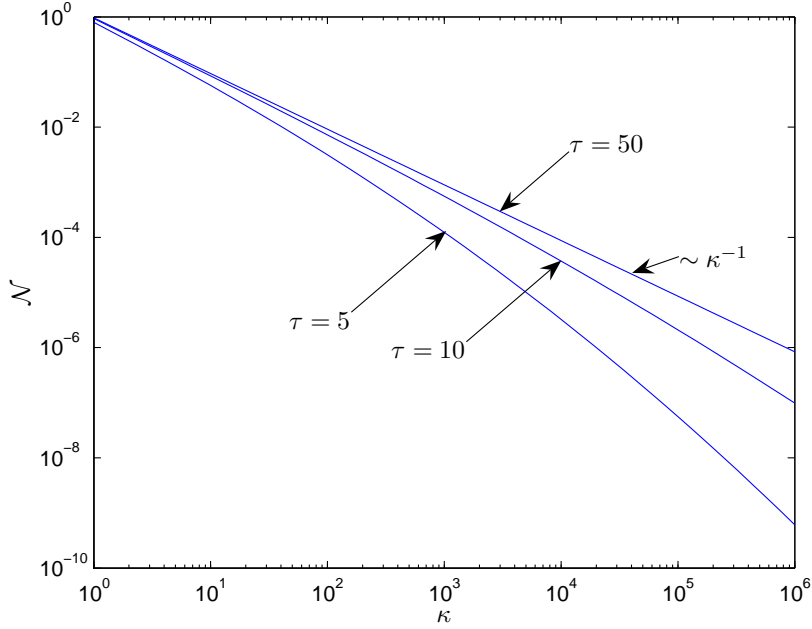}}
\caption{\label{FigDensity} The figure quantifies the deviations from Zipf's law resulting from the finite age of 
the economy, by showing the mean number $N(\kappa,\tau)$ of firms
of normalized size $s/s_0(t)$ larger than $\kappa$ as a function of  $\kappa$, for parameters
$\mu=c_0$, $d+h=0$ (satisfying the balance condition), and for $s_0=100 \cdot s_{\rm min}$,
$c_0=c_1=0$ and reduced times $\tau:= \sigma^2 \theta /2=5;10;50$. The exact asymptotic Zipf's law $\sim \kappa^{-1}$ is also shown for comparison.}
\end{figure}

Figure \ref{FigDensity} shows the mean cumulative number $N(\kappa, \tau)$ of firms as a function of the normalized firm size $\kappa$, for $\mu=c_0$ and $h=-d>0$ satisfying to the balance condition of corollary~\ref{corol1}, for $s_0=100 \cdot s_{\rm min}$, $c_0=c_1=0$ and reduced times $\tau=5$, $10$, $50$. As expected, the older the economy, the closer is the mean cumulative number $N(\kappa,\tau)$ to Zipf's law $N(\kappa, \infty)\sim \kappa^{-1}$. 
Beyond $\tau = 50$, there are no noticeable difference between the actual distribution of firm sizes and its asymptotic power law counterpart. This illustrates graphically the last point discussed in remark~\ref{rem:4} that, the larger the volatility (beyond some threshold), the faster the convergence of the size distribution toward the asymptotic power law. Indeed, the larger the volatility, the smaller the age $\theta$ necessary to reach a value of $\tau$ close to $50$.

The downward curvatures of the graphs for all finite $\tau$'s show that the apparent tail index can be empirically found larger than $1$ even if all conditions for the asymptotic validity of Zipf's law hold. This effect could provide an explanation for 
some dissenting views in the literature about Zipf's law. The two recent influential studies by \citeasnoun{CabralMata2003} and \citeasnoun{Eeckhout2004}\footnote{See the comment by \citeasnoun{Levy09} which suggests that the extreme tail of the size distribution is indeed a power law and the reply by \citeasnoun{Eeckhout09}.}% and the confirmation by \citeasnoun{Yannicketal_city_AER}.}
have suggested that the distribution of firm and of city sizes could be well-approached by the log-normal distribution, which exhibits a downward curvature in a double-logarithmic scale often used to qualify a power law. Our model shows that a slight downward curvature
can easily be explained by the partial convergence of the distribution of firm sizes toward the asymptotic Zipf's law due to the finite age of the economy. 

It is interesting to note that two opposing effects can combine to make the apparent exponent $m$ close to $1$ even when the balance condition does not hold exactly. Consider the situation where $\epsilon:={d+h \over \mu-c_0} < 1$. For $\epsilon <1$, figure \ref{FigEpsilon} shows that $m$ is always less than one. But, figure \ref{FigDensity} shows 
that the distribution of firm sizes for a finite economy is approximately a power law but with an exponent larger than one for the asymptotic regime of an infinitely old economy. It is possible that these two deviations may cancel out to a large degree, providing a nice apparent empirical Zipf's law.

\subsection{Representativeness of the mean-distribution of firm size}

All our results have been established for the average number $N(s,t)$ of firms whose size is larger than $s$, where
the average is performed over an ensemble of equivalent statistical realizations of the economy. Since
empirical data are usually sampled from a single economy, it is important to ascertain if the average
Zipf's law accurately describes the distribution of single typical economies. The answer to this question
is provided by the following proposition whose proof is given in appendix~\ref{app3}.
\begin{proposition}
\label{prop3}
Under assumptions \ref{assumption1} and \ref{assumption:iii-a}, the random number $\tilde N(s,t)$ of firms whose size is larger than $s$ in a given economy follows a Poisson law with parameter $N(s,t)$ (defined in (\ref{ghtgbtpgbqm}) with (\ref{thgdwesa}) and (\ref{dens c<0 expr})):
\be
\label{thidf;wldw}
\Pr \[\tilde N(s,t) = n\] = \frac{N(s,t)^n}{n!} e^{-N(s,t)}.
\ee
\end{proposition}

As a consequence of proposition~\ref{prop3}, we state
\begin{corollary}
\label{corol2}
Under the assumptions of proposition~\ref{prop3}, the variance of the average relative distance 
$\frac{\tilde N(s,t)}{ N(s,t)} -1$ between the number of firms in one realization and its statistical
average is given by
\be
\E \[\(\frac{\tilde N(s,t)}{ N(s,t)} -1 \)^2\] = \frac{1}{N(s,t)}.
\label{thnhwrtadfqf}
\ee
\end{corollary}
\begin{proof}
The left hand side of the equation above is nothing but the variance of $\tilde N(s,t)$ divided by $N(s,t)^2$. Since, $\tilde N(s,t)$ follows a Poisson law, $\Var \tilde N(s,t) = N(s,t)$, hence the result.
\end{proof}

To give a quantitative illustration, let us consider firms
whose sizes evolve according to the pure Geometric
Brownian Motion, i.e., $\Pr[\tilde s_0 = s_0]=1$, $c_0=c_1=h=d=0$ and no minimum exit size. Then,
$N(s,t) =  N(s) = \int_s^\infty g(s') d s'$, where
\begin{equation}
g(s) = \frac{\nu_0}{ \left|\mu - \frac{\sigma^2}{2} \right|}\, s_0^{1-\frac{2\mu}{\sigma^2}}\, s^{\frac{2\mu}{\sigma^2}-2} ~ ,
\qquad s>s_0~ , \qquad \mu < \frac{\sigma^2}{2}~ .
\end{equation}
This expression derives from the general expression (\ref{gstlimit}) for Zipf's law given in appendix~\ref{app:prop1}
in the limit $s_{\rm min} \to 0$. This leads to
\begin{equation}
\label{eqo_sdfgwe}
N(s) = N_0 \, \left({s_0 \over s}\right)^{1-\frac{2\mu}{\sigma^2}}~,
\end{equation}
where
\begin{equation}
\label{eqp_sdfgwe}
N_0 = \int_{s_0}^\infty g(s) ds = \frac{\sigma^2}{2} \cdot \frac{\nu_0}{(\mu - \frac{\sigma^2}{2})^2}  
\end{equation}
is the mean number of firms, whose sizes, at a given time $t$, are larger than
the initial size $s_0$.
Using (\ref{thnhwrtadfqf}) for the variance of the relative distance between the number 
of firms in one realization and its statistical average in
Corollary 2, and with (\ref{eqo_sdfgwe}),  we obtain that the variance of the 
relative distance is given by
$ {1 \over N_0} ~ {s \over s_0}$, where we have assumed that $\mu=0$, so that Zipf's law $N(s)\sim
s^{-1}$ holds for the mean distribution of firm sizes.

In this illustrative example, the total number of firms is infinite while 
$N_0$ remains finite. Let us consider a data set spanning 
the range $s \in (s_0,s_*)$ where $s_* = 0.01\, N_0\, s_0$ is such that
the variance of the relative distance between the number of firms in one realization and its statistical
average remains smaller than $10^{-2}$
over the range $s \in (s_0,s_*)$. Suppose that 
the mean number of
firms in the economy, whose sizes are larger than $s_0$, is equal to $N_0=10^6$. Then
$s_* = 10^4\, s_0$, showing that Zipf's law should be observed, in a single realization
of an economy, with good accuracy over four orders of magnitudes in this example. Figure~\ref{FigSimulation} depicts ten simulation results obtained for such an economy.

\begin{figure}
\centerline{\includegraphics[width=0.75\textwidth]{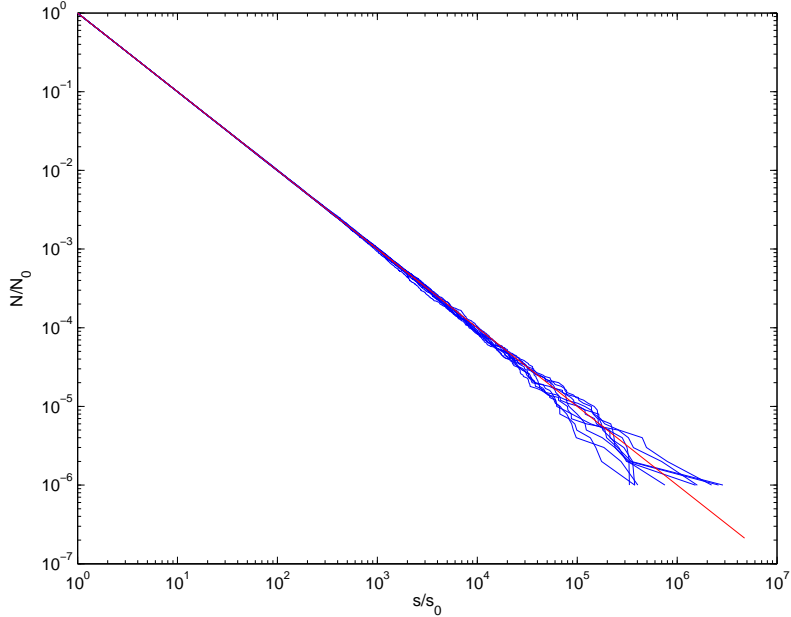}}
\caption{\label{FigSimulation} Number of firms whose size is larger than $s$ when $\sigma=0.01$, $\nu_0=50$ and $\mu=0$ for ten realizations of the economy. The straight red line depicts Zipf's law for the {\em mean} number of firms.}
\end{figure}

\section{Conclusion}

We have presented a general theoretical derivation of Zipf's law, which states 
that, for  most countries, the size distribution of firms is a power law with a specific exponent
equal to $1$: the number of firms with size greater than $S$ is inversely proportional to $S$. 
Our framework has taken into account
time-varying firm creation, firms' exit resulting from both a lack of sufficient size and sudden
external shocks, and Gibrat's law of proportional growth. We have identified
that four key parameters control the tail index $m$ of the power law distribution of firms sizes:
the expected growth rate $\mu$ of incumbent firms,  the hazard rate $h$
of random exits of firms of any size, the growth rate $c_0$ of the size of entrant firms, 
and the growth rate $d$ of the number of new firms. We have identified
that Zipf's law holds exactly when a balance condition holds, namely when
the growth rate $d+c_0$ of investments in new entrant firms is equal to the
average growth rate $\mu-h$ of incumbent firms. Thus, Zipf's law can be
interpreted as a remarkable statistical signature of the long-term optimal
allocation of resources that ensures the maximum sustainable growth rate
of an economy. We have also found that Zipf's
law is recovered approximately when the volatility of the growth rate of individual firms
becomes very large, even when the balance condition does not hold exactly.
We have studied the deviations from Zipf's law due to the finite age of the economy
and shown that a deviation of the balance condition $d+c_0=\mu-h$ can be
compensated approximately by the effect of the finite age of the economy
to give again an approximate Zipf's law. We have also shown that 
the presence of a minimum size below which firms exit 
allows us to account for the age-dependent hazard rate documented
in the empirical literature. Our results hold not only for statistical averages
over ensemble of economies (i.e., in expectations) but also apply to a single typical economy,
as the variance of the relative difference between the number of firms in one realization
and its statistical average decays as the inverse of the number of firms and thus 
goes to zero very fast for sufficiently large economies.
Therefore, our results can be compared with empirical data which are usually sampled
for a single economy. Our theory improves significantly on previous works 
by getting rid of many constraints and conditions that are found unnecessary or
artificial, when taking into account the proper interplay between birth, death and growth.

\appendix
\section{Appendix}

\subsection{Derivation of the distribution of firms' sizes: proof of proposition~\ref{prop1} \label{app:prop1}}

Consider an economy with many firms born at random times $t_i \ge t_0$, $i \in \mathbb N$, where $t_0$ is the starting time of the economy. We assume that no two firms are born at the same time so that the random sequence $\{ t_i \}_{i \in \mathbb N}$ defines a {\em simple  point process} \cite[def. 3.3.II]{DVJ2007}.

Let $S_i(t)$, $i \in \mathbb N$, $t \ge t_0$ be a positive real-valued stochastic process representing the size, at time $t$, of the firm born at $t_i$. Obviously, $S_i(t)=0$, $\forall t < t_i$. The sequence $\{t_i, S_i(t)\}_{i \in \mathbb N}$ defines a {\em simple marked point process} \cite[def. 6.4.I - 6.4.II]{DVJ2007} with ground process $\{ t_i \}_{i \in \mathbb N}$ and marks $\{S_i(t)\}_{i \in \mathbb N}$. 
We assume that $\{ t_i \}$ and $\{S_i(t)\}$ are mutually independent and such that the distribution of
$S_i(t)$ depends only on the corresponding location in time $t_i$. Consequently, the {\em mark kernel} $F_{m,i}(s,t):= \Pr \[S_i(t) < s \]$ simplifies to $F_m\(s,t | t_i\)$.

For any subset $T \times \Sigma$ of $[t_0, \infty) \times {\mathbb R}_+$, we introduce the {\em counting measure}
\bea
N_t\(T \times \Sigma\) &:=& \# \left\{t_i \in T,~ S_i(t) \in \Sigma \right\},\\
&=&  \sum_{i \in \mathbb N:~ t_i \in T} 1_{S_i(t) \in \Sigma}~ .
\eea
The total number of firms whose sizes are larger than $s$ at time $t$ then reads
\bea
\tilde N(s,t) &:=& N_t \( [t_0,t) \times [s, \infty) \),\\
&=& \int_{[t_0,t) \times [s, \infty)} N_t(du \times ds),\\
&=&  \sum_{i \in \mathbb N:~ t_i \le t} 1_{S_i(t) \ge s}~ .
\eea
As a consequence of theorem 6.4.IV.c in \citeasnoun{DVJ2007} we can state that
\begin{lemma}
\label{lemma1}
Provided that the ground process $\{t_i\}_{i \in \mathbb N}$ admits a first order moment measure with density $\nu(t)$ w.r.t Lebesgue measure, the counting process $\tilde N(s,t)$ admits a first moment 
\bea
\label{ghtgbtpgbqm}
N(s,t)&:=& E\[\tilde N(s,t)\],\\
&=& \int_{t_0}^t \[1-F_m\(s,t | u\)\] \cdot \nu(u)\, du.
\eea
\end{lemma}
\begin{remark}
When the ground process is an (inhomogeneous) Poisson process, $\nu(t)$ is nothing but the intensity of the process.
\end{remark}
\begin{proof}
By theorem 6.4.IV.c in \citeasnoun{DVJ2007}, the first-moment measure $M_1 (\cdot) := \E \[N_t(\cdot)\]$ of the marked point process $\{t_i, S_i(t)\}_{i \in \mathbb N}$ exists since the corresponding moment measure exists for the ground process $\{t_i\}_{i \in \mathbb N}$. It reads
\be
M_1 \( du \times ds\) = \nu(u) du \cdot F_m(ds,t | u).
\ee
As a consequence
\bea
N(s,t) &=& \E \[\int_{[t_0,t) \times [s, \infty)} N_t(du \times ds) \]= \int_{[t_0,t) \times [s, \infty)} M_1 \( du \times ds\),\\
&=& \int_{[t_0,t) \times [s, \infty)} \nu(u) du \cdot F_m(ds,t | u) = \int_{t_0}^t \[1-F_m\(s,t | u\)\] \cdot \nu(u)\, du.
\eea
\end{proof}
As an immediate consequence, provided that $S(t)$ admits a density $f_m(s,t|u)$ with respect to Lebesgue measure, the counting process $\tilde N(s,t)$ admits a first-moment density
\be
g(s,t):= \int_{t_0}^t f_m\(s,t | u\) \cdot \nu(u)\, du.
\label{thgdwesa}
\ee
This first-moment density does not sum up to one but to a value $N_\infty(t) =
\lim_{s \to 0} N(s,t)$, which remains finite for all finite $t$. A sufficient condition is that
the growths of the number of firms and of their sizes are not faster than 
exponential in time, in agreement with condition ({\it \i\i}) in proposition~\ref{prop1}.
Many faster-than-exponential growth processes of the number of firms
and of their sizes are also permitted, as long as they do not lead to finite-time
singularities.

\begin{lemma}
\label{lemma2}
Under the assumptions \ref{assumption1}, \ref{assumption:iii-a} and \ref{assumption:5}, the first-moment density of sizes of all the firms existing at the current time $t$ reads
\begin{equation}\label{meandensgeneral}
g(s,t) = \int_{t_0}^t \nu(u) e^{-h \cdot(t-u)} f(s,t|u) du~ , \qquad t>t_0~ ,
\end{equation}
where $t_0 (> t_*)$ is the starting time of the economy (with $t_*$ given by (\ref{eq:birthdate})) and $f(s,t|u)$ is the probability density function of a firm's size at time $t$ and born at time $u$.
\end{lemma}
\begin{proof}
Assumptions \ref{assumption1} and \ref{assumption:iii-a} are enough for lemma~\ref{lemma1} to hold. Besides, by assumption \ref{assumption:5}, the exit rate of a firm is independent from its size so that $f_m(s,t|u) = e^{-h(t-u)} \cdot f(s,t|u)$, where $f(s,t|u)$ denotes the probability density function of a firm's size at time $t$ and born at time $u$.
\end{proof}

Lemmas~\ref{lemma1} and~\ref{lemma2} show that, in order to derive proposition~\ref{prop1}, we just need to consider the law
of a single firm's size, given that it has not yet crossed the level $s_{\min}(t)$. 
The density of a single firm's size, that is solution to equation (\ref{jhojgfwv}) embodying Gibrat's law, for a firm born at time $t_i=t-\theta_i$ and given the condition that the firm's size $S_i(t,\theta_i)$ is larger than $ s_{\min}(t),~ \forall \theta_i  \ge 0$, is given by the following result.
\begin{lemma}
\label{lemma3}
Under the assumptions \ref{assumption:iii-a}, \ref{assumption:Gibrat} and \ref{assumption:c1}, the probability density function $f\(s,t|t-\theta,\tilde s_0=s_0\)$ of a firm's size at time $t$ and aged $\theta$ conditional on $\tilde s_0=s_0$, 
taking into account the condition that the firm would die 
if its size would reach the exit level $s_{\min}(t)$, is
\begin{equation}\label{fphirelpartic}
\begin{array}{c} \displaystyle
f\(s,t|t-\theta,\tilde s_0 = s_0\) = \frac{1}{2\sqrt{ \pi\tau} s} \[ \exp\left(-\frac{1}{4 \tau} \(\ln\left(\frac{s}{s_{\min}(t)}\right)-
\ln \(\frac{s_0(t)}{s_{\min}(t)}\) - (\delta-1-\delta_0)\tau \)^2 \right)- \right.\\[4mm]
\displaystyle \left. \(\frac{s_0(t)}{s_{\min}(t)}\)^{- \(\delta-1 - \delta_0\)} 
\left(\frac{s}{s_{\min}(t)}\right)^{\delta_0 -\delta_1}
\exp\left(-\frac{1}{4 \tau} \(\ln\left(\frac{s}{s_{\min}(t)}\right) + 
\ln \(\frac{s_0(t)}{s_{\min}(t)}\) - (\delta-1-\delta_0)\tau \)^2 \right)\]~ ,
\end{array}
\end{equation}
where
\begin{equation}\label{lambzero}
s_0(t):=s_0 e^{c_0 \cdot t}~, \qquad  \tau := \frac{\sigma^2}{2}\, \theta~ ,\qquad 
\delta :={2 \mu \over \sigma^2}~,  \qquad 
\delta_0 :={2 c_0 \over \sigma^2}~,  \qquad \delta_1 :={2 c_1 \over \sigma^2}~.
\end{equation}
\end{lemma}

\begin{proof}
Let us consider a firm born at time $u=t-\theta$, where $t$ denotes the current time and $\theta \ge 0$ is the age of the firm. The firm's size $S(\theta, u)$ is given by the following stochastic process
\begin{equation}\label{gbmS}
S(\theta,u)= s_0(u) e^{c \cdot \theta + \sigma W(\theta)}~,
\end{equation}
where $\theta=t-u$,
$W(\theta)$ is a standard Wiener process, while $s_0(u)$ is the initial size of the firm, given $\tilde s_0 = s_0$, and $c:= \mu-\frac{\sigma^2}{2}$. The process (\ref{gbmS}) with the initial and boundary conditions in assumptions \ref{assumption:iii-a} and \ref{assumption:c1} can be reformulated as
\begin{equation}
\label{sthrusone}
S(\theta, u) = s_{\min}(u+\theta) e^{\mathcal{Z}(\theta,u)}~ ,
\end{equation}
where
\begin{equation}\label{mathzdef}
\mathcal{Z}(\theta,u) = \ln \rho(u+\theta) +(c-c_1) \theta + \sigma W(\theta), \qquad \rho(t) := \frac{s_0(t)}{s_{\min}(t)} ~.
\end{equation}
As a consequence,
\begin{equation} \label{eq:sdmmosit}
f\(s,t|u,\tilde s_0 = s_0\) = \frac{1}{s} \varphi\left[\ln\left(\frac{s}{s_{\min}(t)}\right),\theta;u\right] ~,
\end{equation}
where $\varphi(z ;\theta, u)$ denotes the density of $\mathcal{Z}(\theta,u)$ which is solution to 
\begin{equation}\label{difeqz}
\begin{array}{c} \displaystyle
{\partial \varphi(z;\theta,u) \over \partial \theta} + (c-c_1) {\partial
\varphi(z;\theta,u) \over
\partial z} = {\sigma^2 \over 2} {\partial^2 \varphi(z;\theta, u) \over
\partial z^2}~ , \\[4mm]
\varphi(z;\theta=0, u)= \delta(z-\ln\rho(u))~ , \\[3mm]
\varphi(z=0;\theta, u) =0~ , \qquad \theta>0~ .
\end{array}
\end{equation}
These initial and boundary conditions are equivalent to the initial and boundary conditions in assumptions \ref{assumption:iii-a} and \ref{assumption:c1}.
Using any textbook on stochastic processes \cite[for instance]{Redner01}, we get
\begin{equation}\label{solpsievent}
\begin{array}{c} \displaystyle
\varphi(z;\theta, u) = {1 \over 2\sqrt{ \pi\tau}}  \exp\left(-{(z-
\ln\rho(u)- (\delta -1-\delta_1)\tau)^2 \over 4
\tau} \right)- \\[4mm] \displaystyle { \left[\rho(u)\right]^{\delta_1-\delta+1} \over 2\sqrt{ \pi\tau}} 
\exp\left(-{(z+ \ln\rho(u) - (\delta-1-\delta_1)\tau)^2 \over 4 \tau} \right)~ ,
\end{array}
\end{equation}
where $\delta$ and $\tau$ are defined in \eqref{lambzero}.
Taking into account the relation
\begin{equation}
\rho(u) = \rho(t) e^{(\delta_1 -\delta_0)\tau}~,
\end{equation}
we rewrite expression (\ref{solpsievent}) as
\begin{equation}\label{solpsieventqeg}
\begin{array}{c} \displaystyle
\varphi(z;\theta, u) = {1 \over 2\sqrt{ \pi\tau}} \exp\left(-{(z-
\ln\rho(t)- (\delta -1-\delta_0)\tau)^2 \over 4
\tau} \right)- \\[4mm] \displaystyle { \left[\rho(t)\right]^{\delta_0-\delta+1} \over 2\sqrt{ \pi\tau}} 
\exp\left(-{(z+ \ln\rho(t) - (\delta-1-\delta_0)\tau)^2 \over 4 \tau} +(\delta_0 -\delta_1) z \right)~ ,
\end{array}
\end{equation}
By substitution in \eqref{eq:sdmmosit}, this concludes the proof of Lemma 3.
\end{proof}

Performing the change of variable from birthdate $u$ to 
age $\theta=t-u$ in (\ref{meandensgeneral}), and accounting for assumption~\ref{assumption1}, i.e. the fact that
$ \nu(t) = \nu_0\, e^{d\cdot t}~$, leads to
\begin{equation}\label{meandenthetaint0}
g(s,t) = \nu(t) \int_0^{\theta_0} e^{-(d+h)\theta} \E \[f\(s,t|t-\theta,\tilde s_0\)\] d\theta~ ,
\end{equation}
where $\theta_0 = t-t_0$ is the age of the given economy.  $\E \[f\(s;t,|t-\theta,\tilde s_0\)\]$
denotes the statistical average of $f\(s;t,\theta|\tilde s_0\)$ over the random variable $\tilde s_0$. Inasmuch as
$t_0$ should not be smaller than $t_*$ given by \eqref{eq:birthdate}, 
we should thus have $\theta_0< \theta_* := \frac{\ln \rho(t)}{c_0-c_1}$.

As a byproduct, the mean density of firm sizes, conditional on $\tilde s_0=s_0$ is
\begin{equation}\label{meandenthetaint}
g\(s,t|\tilde s_0=s_0\) = \nu(t) \int_0^{\theta_0} e^{-(d+h)\theta} f\(s,t|t-\theta,\tilde s_0=s_0\) d\theta~ .
\end{equation}
Thus, substituting \eqref{fphirelpartic} into \eqref{meandenthetaint} yields
\begin{equation}\label{gthruGtos}
g(s,t|\tilde s_0=s_0) = \frac{\tilde{\nu}(t)}{s}\, G\left(\ln\left({s \over
s_{\min}(t)}\right);t,\tau_0\right)~ , \qquad \tilde{\nu}(t) = {2 \nu(t)
\over \sigma^2 }~ ,
\end{equation}
with
\begin{equation}\label{Gzdef}
G(z;t,\tau_0) := \int_0^{\tau_0} e^{-\eta \tau} \varphi(z;t,\tau) d\tau~ ,
\end{equation}
where $\varphi(z;t,\theta)$ is given by \eqref{solpsievent} while 
\begin{equation}
\tau_0 := {\sigma^2 \over 2}\, \theta_0 \qquad (\tau_0<\tau_*)~ , \qquad \eta :=
\frac{2}{\sigma^2}\,(d+h)~ .
\end{equation}
The substitution of $\varphi(z;t,\theta)$ from \eqref{solpsievent} into the integral \eqref{Gzdef} leads to two integrals, which can be reduced to
\begin{equation}\label{tableinttheta}
\mathcal{I}(z,\theta,\alpha, \beta) := \int_0^\theta \exp\left(-{(z-\alpha
\tau)^2 \over 4 \tau}- \beta \tau\right) {d\tau \over 2\sqrt{\pi \tau}}~ ,
\end{equation}
whose expression can be obtained by the tabulated integral (7.4.33) in \citeasnoun{AB1965} by the change of variable $u=\sqrt{\tau}$. This leads to 
\begin{equation}\label{exprGtaustze}
\begin{array}{c} \displaystyle
G(z;t,\tau_0) = {1 \over 2\alpha(\eta)} \times \\[4mm]
\displaystyle \Bigg\{ e^{{1 \over 2} \left(\alpha z_- -
\alpha(\eta)|z_-|\right)} \text{erfc} \left( {|z_-| - \tau_0 \alpha(\eta) \over
2 \sqrt{\tau_0}}\right) - e^{{1 \over 2} \left(\alpha z_- + \alpha(\eta)
|z_-|\right)} \text{erfc} \left(
{|z_-| + \tau_0 \alpha(\eta) \over 2 \sqrt{\tau_0}}\right)- \\[4mm] \displaystyle
\rho(t)^{-\alpha} \Big[e^{{1 \over 2} \left(\alpha z_+ - \alpha(\eta) |z_+|
\right)} \text{erfc} \left( {|z_+| - \tau_0 \alpha(\eta) \over 2
\sqrt{\tau_0}}\right) - e^{{1 \over 2} \left(\alpha z_+ + \alpha(\eta)
|z_+|\right)} \text{erfc} \left( {|z_+| + \tau_0 \alpha(\eta) \over 2
\sqrt{\tau_0}}\right) \Big]\Bigg\}~ ,
\end{array}
\end{equation}
with
\begin{equation}\label{skappa def}
\alpha:= \delta - 1 - \delta_0~, \qquad \alpha(\eta) := \sqrt{\alpha^2 + 4 \eta}~ , \qquad z_- := \ln \frac{s}{s_0(t)}~ , \qquad z_+ := \ln \frac{s \cdot s_0(t)}{s_{\min}(t)^2} ~.
\end{equation}

For an old enough economy, i.e., when $\sqrt{\tau_0} \gg 1/\alpha(\eta)$, we can expand expression \eqref{exprGtaustze} to obtain
\begin{equation}\label{Gtauinf}
G_\infty(z;t) = {1 \over \alpha(\eta)} \left[e^{{1 \over 2} \left(\alpha z_- -
\alpha(\eta)|z_-|\right)} -\rho(t)^{-\alpha} e^{{1 \over 2} \left(\alpha z_+ -
\alpha(\eta) |z_+| \right)} \right]~ .
\end{equation}
Substituting this last expression into equation \eqref{gthruGtos} for the mean
density of firms sizes, and after making explicit the $s$-dependence of
the variable $z$, we finally get 
\begin{equation}\label{gstlimit}
g(s,t|\tilde s_0=s_0) = {\tilde{\nu}(t) \over s \alpha(\eta)}\,
\begin{cases}
\(\frac{s}{s_0(t)}\)^{{1 \over 2}(\alpha-\alpha(\eta))} \left(1-\(\frac{s_0(t)}{s_{\min}(t)}\)^{-\alpha(\eta)} \right)
~ , & s > s_0(t)~ , \\[3mm]
\(\frac{s}{s_0(t)}\)^{{1 \over 2} (\alpha+\alpha(\eta))} - \(\frac{s_0(t)}{s_{\min}(t)}\)^{-\alpha(\eta)} \(\frac{s}{s_0(t)}\)^{{1 \over 2} (\alpha-\alpha(\eta))}~ , &  s_0(t)> s > s_{\min}(t)~ .
\end{cases}
\end{equation}
for large $\tau_0\gg \alpha(\eta)^{-1}$, with $s_0(t)=s_0 e^{c_0 \cdot t}$, as defined by \eqref{lambzero}.

According to assumption~\ref{assumption:iii-a}, the expectation of $g(s,t|\tilde s_0)$ with respect to $\tilde s_0$ provides us with the unconditional mean density of firm sizes
\be
\label{dens c<0 expr}
g(s,t) \approx {\tilde{\nu}(t) \over s \alpha(\eta)} \cdot \( \frac{\E\[\tilde s_0^m\]^{1/m} e^{c_0 \cdot t}}{s}\)^m, \qquad \text{as } s \to \infty~~{\rm and}~ t \to \infty~,
\ee
where $m$ is given by (\ref{eq:m}).
This expression (\ref{dens c<0 expr}) justifies the statement of proposition~\ref{prop1} and concludes the proof.

\subsection{Growth rate of the overall economy: Proof of proposition~\ref{prop2} \label{app2}}

Using the same machinery as in appendix~\ref{app:prop1}, we define the total size of the economy at time $t$ as
\bea
\tilde \Omega(t) &:=& \sum_{i \in \mathbb N:~ t_i \le t} S_i(t) \nonumber \\
&=& \int_{t_0}^{t} s \cdot N_t(du \times ds)~ .
\eea
Under the assumptions of proposition~\ref{prop1}, by theorem 6.4.V.iii in \citeasnoun{DVJ2007}, we get
\begin{eqnarray}
\Omega(t) &:=& \E \[\tilde \Omega(t)\]  \label{eq:lkesu} \\
&=& \nu(t) \int_0^{\tau_0} e^{-\eta \cdot \tau} \E\[S(t,\tau)\] d\tau~ . \label{eq:dlkhug}
\end{eqnarray}
For simplicity, let us consider the case where $s_{\min}=0$. This assumption is not necessary, but greatly simplifies the calculation. Under this assumption, the size of an incumbent firm follows a geometric Brownian motion so that
\be
\label{eq:mkdsh}
\E\[S(t,\tau)\] = s_0(t) e^{(\delta - \delta_0) \tau}~,
\ee
where $\delta$, $\delta_0$ and $s_0(t)$ are defined in \eqref{lambzero}.
Substituting (\ref{eq:mkdsh}) into \eqref{eq:dlkhug} gives
\bea
\Omega(t) &=& \nu(t) \cdot s_0(t) \int_0^t e^{(\mu-c_0 - h-d) u}  du~,\\
&=& \int_0^t e^{(\mu-h) \cdot (t-u)} \nu(u) \cdot s_0(u) du~,\\
&=& \int_0^t e^ {(\mu-h) \cdot (t-u)} dI(u)~. \label{eq:mesurhg}
\eea
This last equation shows that $\mu-h$ is the return on investment of the economy. By integration, we get the limit growth rate of the economy
\be
\lim_{t \to \infty} \frac{d \ln \Omega(t)}{dt} =
\begin{cases}
\mu - h~, & \mu - h > d + c_0 \\[3mm]
d + c_0~, & \mu - h \le d + c_0
\end{cases}~ .
\ee
This concludes the proof of proposition~\ref{prop2} when $s_{\min}=0$. When $s_{\min} \neq 0$ and grows at the rate $c_1 \geq 0$, the result can still be proved along the same lines but at the price of more tedious calculations since the expectation in \eqref{eq:mkdsh} involves eight error functions.

\subsection{Representativeness of the mean-distribution of firm size: Proof of proposition~\ref{prop3} \label{app3}}

The proof of proposition~\ref{prop3} follows from lemma 6.4.VI in \citeasnoun{DVJ2007} which states that a marked point process that has mark kernel $F_m(s,t|u)$, and for which the Poisson ground process has intensity measure $\nu(u) du$, is equivalent to a Poisson process on the product space with intensity measure $\Lambda(du \times ds) = F_m(ds,t|u) \cdot \nu(u) du$.

Thus, under the assumptions \ref{assumption1} and \ref{assumption:iii-a}, using the notations of appendix~\ref{app:prop1}, the marked point process $\{t_i, S_i(t)\}_{i \in \mathbb N}$ is a compound Poisson process with intensity measure $\Lambda (du \times ds)$. Consequently
\bea
\Pr \[ \tilde N(s,t) = n \] &=& \Pr \[ N_t \( [t_0,t) \times [s, \infty) \) = n \]~, \nonumber \\
&=& \frac{\(\int_{[t_0,t) \times [s, \infty)} F_m(ds,t|u) \cdot \nu(u) du\)^n}{n!} \exp\(-\int_{[t_0,t) \times [s, \infty)} F_m(ds,t|u) \cdot \nu(u) du\)~,   \nonumber \\
&=&\frac{N(t,s)^n}{n!} \cdot e^{-N(s,t)},
\eea
by lemma~\ref{lemma1}.

\end{document}